\documentclass[a4paper,12pt,notitlepage,accepted=2022-03-23,onecolumn]{quantumarticle}
\pdfoutput=1
\usepackage{algorithmic}
\usepackage[top=0.75in,bottom=0.75in,left=0.75in,right=0.75in]{geometry}
\usepackage{tabularx}
\usepackage{amsfonts}
\usepackage{amssymb}
\usepackage{amsthm}
\usepackage{amsmath}
\usepackage{amsthm}
\usepackage{amsxtra}
\usepackage{color}
\usepackage{xcolor}
\usepackage{authblk}
\usepackage[framemethod=tikz]{mdframed}
\usepackage[utf8]{inputenc} 
\usepackage[T1]{fontenc}
\usepackage{comment}
\usepackage{graphicx}
\usepackage{enumerate}
\usepackage[caption=false]{subfig}
\newtheorem{theorem}{Theorem}[section]
\newtheorem{lemma}[theorem]{Lemma}

\usepackage[ colorlinks = true,
linkcolor = blue,
urlcolor  = blue,
citecolor = blue,
anchorcolor = green,
]{hyperref}

\def\bra #1{\langle #1\vert}
\def\ket #1{\vert #1\rangle}

\def\tr{{\rm Tr}}

\newcommand*{\ExpE}{\mathbb{E}}

\newcommand*{\cE}{\mathcal{E}}

\newcommand*{\cG}{\mathcal{G}}

\newcommand*{\cT}{\mathcal{T}}

\newcommand{\bc}{\begin{center}}
\newcommand{\ec}{\end{center}}

\def\01{\{0,1\}}

\newcommand{\proj}[1]{|#1\rangle\langle#1|}

\newcommand*{\level}{p}
\usepackage{mathtools}
\mathtoolsset{showonlyrefs}

\author[1]{Sergey Bravyi}
\author[2]{Alexander Kliesch}
\author[3]{Robert Koenig}
\author[4]{Eugene Tang}
\affil[1]{IBM Quantum, IBM T.J. Watson Research Center, Yorktown Heights, NY 10598, USA}
\affil[2]{Zentrum Mathematik, Technical University of Munich, 85748 Garching, Germany}
\affil[3]{Institute for Advanced Study \& Zentrum Mathematik, Technical University of Munich, \\ 85748 Garching, Germany}
\affil[4]{Institute for Quantum Information and Matter, Caltech,
Pasadena, CA 91125}
\date{}

\title{Hybrid quantum-classical algorithms for approximate graph coloring}
\begin{document}

\maketitle
\begin{abstract}
We show how to apply the recursive quantum approximate optimization algorithm (RQAOA) to MAX-$k$-CUT, the problem of finding an approximate vertex $k$-coloring of a graph. We compare this proposal to the best known classical and hybrid classical-quantum algorithms. First, we show that the standard (non-recursive) QAOA fails to solve this optimization problem for most regular bipartite graphs at any constant level~$\level$: the approximation ratio achieved by QAOA is hardly better than assigning colors to vertices at random. Second, we construct an efficient classical simulation algorithm which simulates level-$1$ QAOA and level-$1$ RQAOA  for arbitrary graphs. In particular, these hybrid algorithms give rise to efficient classical algorithms, and no benefit arising from the use of quantum mechanics is to be expected.  Nevertheless, they provide a suitable testbed for assessing the potential benefit of hybrid algorithm: We use the simulation algorithm to perform large-scale simulation of level-$1$ QAOA and RQAOA with up to $300$~qutrits applied to ensembles of randomly generated $3$-colorable constant-degree graphs.  We find that level-$1$ RQAOA is surprisingly competitive: for the ensembles considered, its approximation ratios are often higher than those achieved by the best known generic classical algorithm based on rounding an SDP relaxation. This suggests the intriguing possibility that higher-level RQAOA may be a potentially useful algorithm for NISQ devices.

\end{abstract}

\section{Introduction\label{sec:intro}}
Combinatorial optimization is currently considered one of the most promising areas of application of near-term quantum devices. One of the primary restrictions of such devices is the fact that they may only execute  short-depth circuits with reasonable fidelity. This is a consequence of the lack of sophisticated fault-tolerance mechanisms. Variational quantum algorithms such as the quantum approximate optimization algorithm (QAOA)~\cite{farhiqaoa} can deal with this restriction because parameters such as the circuit depth/architecture can be chosen according to existing experimental restrictions. This is in contrast to more involved quantum algorithms e.g., for algebraic problems: these may involve circuit depths scaling with the problem size and have additional requirements on the connectivity of available inter-qubit operations.

QAOA and similar proposals are hybrid algorithms: Here we envision operations on the quantum device to be supplemented by efficient, that is, polynomial-time classical processing. As a result, the
algorithmic capabilities of such a hybrid setup should be compared to the class of efficient (polynomial-time) classical probabilistic algorithms. This means that corresponding proposals face stiff competition against decades of algorithms research in classical computer science. While there are complexity-theoretic arguments underscoring the power of e.g., constant-depth quantum circuits, the jury is still out on whether hybrid algorithms for near-term devices can indeed provide a provable computational advantage over comparable classical algorithms. 

Recent no-go results  show limitations of variational quantum algorithms for the well-studied MAX-CUT problem: in~\cite{brakoeklietan}, we showed that the  Goemans-Williamson-algorithm -- the best known classical algorithm for this problem --  outperforms
QAOA for any constant level $\level$ (which amounts to constant-depth for bounded-degree graphs)  in terms of the achieved approximation ratio. This result extends to  more general, possibly non-uniform $\mathbb{Z}_2$-symmetric hybrid local quantum algorithms, and shows that corresponding circuits will need a circuit depth growing at least logarithmically with the problem size to yield a better approximatio ratio.  More recently,  Farhi, Gamarnik and Gutmann~\cite{farhi2020quantum} exploited the spatial uniformity of the QAOA algorithm to give an extremely elegant argument demonstrating a similar logarithmic-depth lower bound for QAOA for random $d$-regular graphs. The title of reference~\cite{farhi2020quantum} aptly summarizes this conclusion using the words ``the QAOA needs to see the whole graph''.

To go beyond these negative results, one could attempt to simply use logarithmic-depth circuits with the reasoning that for small to intermediate problem sizes, the corresponding circuit depths may still be amenable to realization by a near-term device. In addition to the problem of fault-tolerance, the potential merits of this idea are unfortunately difficult to assess by means of classical simulation. Indeed, the very idea that this approach yields computational benefits over classical algorithms is in tension with 
efficient classical simulability of corresponding quantum processes.

An alternative to this necessarily limited approach is to try to find new ways of leveraging the information-processing capabilities of short-depth circuits by introduction of non-local (classical) pre- and post-processing steps, thereby sidestepping the key assumption of locality in the aforementioned  no-go results. As long as one restricts to
procedures where these additional processing steps can be executed efficiently (e.g., in polynomial time) and the quantum device is used only a polynomial number of times, the resulting hybrid algorithm
still has the feature of being efficiently executable with near-term hardware. At the same time, it may provide additional descriptive power, ultimately (hopefully) resulting in improved approximation ratios.

The recursive quantum approximate optimization algorithm (RQAOA) 
is a proposal of this kind which makes use of  QAOA (with a constant level~$\level$) in a recursive fashion with the goal of improving approximation ratios. Numerical evidence obtained for MAX-CUT on random graphs  indicates that even level-$1$ RQAOA significantly improves upon QAOA (at the same level)  for the MAX-CUT problem~\cite{brakoeklietan}. While the problem of establishing a rigorous lower bound on the approximation ratio achieved by this hybrid algorithm remains open (except for very special families of instances), this indicates that RQAOA  -- especially at higher levels~$\level$ -- has the potential to yield results competitive with the best known classical algorithms.

Here we consider the MAX-$k$-CUT problem,  an optimization version of the graph $k$-coloring problem. Suppose $G=(V,E)$ is a graph with $n=|V|$ vertices and $e=|E|$ edges. Given an integer $k\geq 2$, the goal is to find an approximate $k$-coloring of vertices of~$G$ which maximizes the number of edges whose endpoints have different colors. For each vertex $j\in V$, introduce a variable $x_j\in \mathbb{Z}_k$ which represents a color assigned to~$j$. The $k$-coloring cost function to be maximized is defined as
\begin{align}
C(x)&=\sum_{(i,j)\in E}(1-\delta_{x_i,x_j})\qquad\textrm{ for }\qquad x\in \mathbb{Z}_k^n\ .\label{eq:figureofmeritmaxkcut}
\end{align}
The performance of a $k$-coloring algorithm on a given graph $G$ is usually quantified by its approximation ratio $\alpha$,
that is, the ratio between the expected value of the cost function~$C(x)$ on a coloring~$x$
produced by the algorithm and the maximum value~$\max_x C(x)$.

The MAX-$k$-CUT problem can also be viewed as an anti-ferromagnetic $k$-state Potts model. The standard MAX-CUT problem corresponds to~$k=2$. The problem is well-studied. Consider first the special case where $G$ is a $k$-colorable graph. Clearly, a uniformly random assignment of colors~$x$ achieves an approximation ratio of~$1-1/k$ on average.
 For the case where~$k$ is a power of two, Cho, Raje and Sarrafzadeh~\cite{ChoRaje98} constructed an $O((e+n)\log k)$-time algorithm which achieves an approximation ratio of $1-1/k(1-1/n)^{\log k}$, improving upon 
a deterministic $O(enk)$-time algorithm~\cite{vitanyi81} achieving the same ratio $1-1/k$ as random coloring (and obtained by derandomizing the latter). In the general case, when $G$ may not be $k$-colorable, Frieze and Jerrum~\cite{friezejerrum} gave an algorithm achieving an approximation ratio~$1-1/k+2\log k/k^2$ for arbitrary sufficiently large~$k$. This is known to be close to optimal since no polynomial-time algorithm can achieve an approximation ratio better than~$1-1/(34k)$ unless  $P=NP$~\cite{kann97}, and is indeed optimal if one assumes the Unique Games Conjecture~\cite{consequencesugc}. Frieze and Jerrum's algorithm is based on an SDP relaxation and a randomized rounding scheme inspired by Goemans and Williamson's algorithm for MAX-CUT~\cite{goeWill95} and comes with detailed estimates of the approximation ratio $\alpha_k$ achieved by the algorithm, namely  
\begin{align}
\alpha_2&\geq 0.878567\\
\alpha_3 &\geq 0.800217\\
\alpha_4&\geq 0.850304\\
\alpha_5&\geq 0.874243
\end{align}
Here the bound on $\alpha_2$ matches the Goemans-Williamson algorithm for MAX-CUT. This was further improved by a new algorithm by Klerk et al.~\cite{klerketal} with a guaranteed approximation ratio of
\begin{align}
\alpha_3&\geq 0.836008\label{eq:alphathreebound}\\
\alpha_4&\geq 0.857487\ .
\end{align}
Their algorithm is based on the properties of the Lovasz $\vartheta$-function and achieves the best currently known approximation ratio (for polynomial-time algorithms) for $k=3$. Goemans and Williamson themselves~\cite{goemansmax3cut} independently introduced an algorithm based on so-called ``complex semidefinite programming'' that matches this ratio for $k=3$. Unfortunately, the analysis involved in establishing the bound~\eqref{eq:alphathreebound} is significantly more complicated than the MAX-CUT case and is not known to be generalizable to arbitrary~$k$. Fortunately, Newman~\cite{newman18} describes a simple rounding procedure that leads to an algorithm provably matching the bound~\eqref{eq:alphathreebound} while only being slightly worse than Frieze/Jerrum for larger values of~$k$. Finally, it should also be noted that dense $3$-colorable graphs can be $3$-colored in (randomized) polynomial time~\cite{alonkahale97}.

Here we study approximation ratios achieved by QAOA and RQAOA, respectively for MAX-$k$-CUT with $k>2$, and compare these to those achieved by the best known classical approximation algorithms for this problem discussed above.

\paragraph{Outline}
In Section~\ref{sec:hybridalgorithms}, we discuss how to formulate QAOA and RQAOA for the MAX-$k$-CUT-problem.  In Section~\ref{sec:upperbound}, we establish limitations on constant-level QAOA.  In Section~\ref{sec:classSimQAOA}, we describe an efficient classical simulation algorithm for simulating level-$1$ QAOA and level-$1$ RQAOA. In Section~\ref{sec:numericalresults} , we discuss our numerical findings comparing level-$1$ QAOA and  level-$1$ RQAOA with the best known classical algorithm for MAX-$k$-CUT.

\section{Hybrid algorithms for MAX-$k$-CUT\label{sec:hybridalgorithms}}
Here we give a brief description of how QAOA and RQAOA can be adapted to apply to MAX-$k$-CUT for any $k\geq 2$.
\subsection{ QAOA to MAX-$k$-CUT}
Solving MAX-$k$-CUT by level-$\level$ QAOA (in the following denoted by QAOA$_\level$) proceeds in an established fashion as in the case of MAX-CUT (i.e., $k=2$). One notable feature is that in the MAX-$k$-CUT problem, it is natural to work with~$n$ qudits of dimension~$k$ each instead of qubits (all quantum circuits considered below can be simulated on a standard qubit-based quantum computer with a constant-factor overhead). We use the $n$-qudit cost function Hamiltonian
\begin{align}
C=\sum_{1 \leq i<j \leq n} \sum_{b \in \mathbb{Z}_{k}} J_{i, j}(b) \Pi_{i, j}(b)\label{eq:costfunctionhamiltonian}
\end{align}
where $J_{i, j}(b)$ are real coefficients and
\begin{align}
\Pi(b)=\sum_{a \in \mathbb{Z}_{k}}|a, a+b\rangle\langle a, a+b| \label{eq:Pibprojectordef}
\end{align}
is a diagonal projector acting on $\mathbb{C}^{k} \otimes \mathbb{C}^{k}$. Here and below the addition of color indices is performed modulo~$k$. The subscripts $i, j$ in $\Pi_{i, j}(b)$ indicate which pair of qudits is acted upon by $\Pi(b)$. By definition, $\Pi_{i,j}(b)=\Pi_{j,i}(-b)$.

Equation~\eqref{eq:costfunctionhamiltonian} defines a general class of cost function Hamiltonians. The MAX-$k$-CUT problem on a graph $G=(V,E)$ associated with the cost function~\eqref{eq:figureofmeritmaxkcut} corresponds to the choice of coefficients 
 \begin{align}
J_{i,j}(b) = 
\begin{cases}
0\ , & (i,j)\notin E \\
1 - \delta_{b, 0}\ , & (i,j) \in E\ .
\end{cases}\label{eq:jijbeq}
\end{align}
The Hamiltonian $C$ defined Eq.~\eqref{eq:costfunctionhamiltonian}
commutes with the symmetry operator $X^{\otimes n}$, where
\begin{align}
X=\sum_{c \in \mathbb{Z}_{k}}|c+1\rangle\langle c|\ 
\end{align}
is the generalized Pauli-$X$ operator. This is analogous to the $\mathbb{Z}_2$-symmetry of the Ising model exploited in~\cite{brakoeklietan}. Motivated by the corresponding ansatz for MAX-CUT, we generalize the level-$\level$ QAOA ansatz to qudits as
\begin{align}
\ket{\psi(\beta,\gamma)}=U(\beta,\gamma)\ket{+^n}\  ,\label{eq:psigammaqaoa}
\end{align}
where the level-$\level$ QAOA unitary is given by
\begin{align}
U(\beta,\gamma)&=\prod_{t=1}^{\level} B(\beta^{(t)})^{\otimes n}e^{-i\gamma^{(t)} C}\ , 
\end{align}
with $\beta=(\beta^{(1)},\ldots,\beta^{(\level)})\in(\mathbb{R}^k)^\level$, $\gamma=(\gamma^{(1)},\ldots,\gamma^{(\level)})\in\mathbb{R}^\level$ and where for $\beta\in\mathbb{R}^k$, the unitary~$B(\beta):\mathbb{C}^k\rightarrow\mathbb{C}^k$ 
is diagonal in the eigenbasis of~$X$ and given by
\begin{align}
B(\beta)&=\sum_{a\in\mathbb{Z}_k}e^{i\beta_a}\proj{\phi_a}\ ,\qquad \ket{\phi_a}\equiv Z^a\ket{+}\ .\label{eq:Deltadef}
\end{align}
In this expression,
\begin{align}
Z&=\sum_{a\in\mathbb{Z}_k} \omega^a \proj{a}\ ,\qquad \omega\equiv e^{2\pi i/k}
\end{align}
is the generalized Pauli-$Z$ operator, whereas $\ket{+}\in\mathbb{C}^k$ is the $+1$ eigenvector of $X$, that is, $\ket{+}=k^{-1/2}\sum_{b\in\mathbb{Z}_k}\ket{k}$. Level-$\level$ QAOA for MAX-$k$-CUT proceeds by
\begin{enumerate}[(i)]
\item \label{it:stepfirstqaoa}
first maximing the expected value $\bra{\psi}C\ket{\psi}$ where $\ket{\psi}\equiv \ket{\psi(\beta,\gamma)}$ over $(\beta,\gamma)\in (\mathbb{R}^k)^\level\times\mathbb{R}^\level$. 
 \item
then measuring an energy-maximizing state $\ket{\psi_*}=\ket{\psi(\beta_*,\gamma_*)}$ in the computational basis.
 \end{enumerate}
The output of this process is a coloring $x\in\mathbb{Z}_k^n$ achieving an expected approximation ratio 
\begin{align}
\max_{\beta,\gamma}\bra{\psi}C\ket{\psi}/\max_x C(x)\ .
\end{align} This concludes the description of QAOA$_\level$.

\subsection{Adapting RQAOA to MAX-$k$-CUT}
The RQAOA proceeds by successively reducing the size of the problem, eliminating a single variable in each step by a procedure called correlation rounding.   To formulate RQAOA for MAX-$k$-CUT, we begin by noting that the family of Hamiltonians defined by~\eqref{eq:costfunctionhamiltonian} is closed under variable eliminations (up to irrelevant additive constants).

Level-$p$ RQAOA proceeds by iterative applications of several single variable elimination steps as described by~\eqref{it:stepfirstqaoa} and~\eqref{it:steptworqaoa} described below. We note that in the very first variable elimination step, the scalar~$M_{i,j}(b)$ (for a fixed pair of vertices $(i,j)$) computed in~\eqref{it:stepfirstqaoa} depends only on whether or not~$b$ is non-zero.  This is a consequence of the specific form~\eqref{eq:jijbeq} of the coefficients~$J_{i,j}(b)$ in the (initial) MAX-$k$-CUT cost function Hamiltonian~$C$ (Eq.~\eqref{eq:costfunctionhamiltonian}). However, this is no longer necessarily the case after one or more variable elimination steps have been completed since the cost function Hamiltonian is updated according to~\eqref{it:steptworqaoa}.

A single variable elimination step of level-$\level$ RQAOA works as follows:
\begin{enumerate}[(i)]
\item\label{it:steponerqaoa}
First, maximize the expected value $\bra{\psi}C\ket{\psi}$ with $\ket{\psi}=\ket{\psi(\beta,\gamma)}$ over $(\beta,\gamma)\in(\mathbb{R}^k)^\level\times \mathbb{R}^\level$. Then compute
the mean value
\begin{align}
M_{i,j}(b)=\bra{\psi}\Pi_{i,j}(b)\ket{\psi}\qquad\textrm{ for all pairs of vertices } (i,j)\ .
\end{align}
Note that $0\leq M_{i,j}(b)\leq 1$ since $\Pi_{i,j}(b)$ is a projector. 

\item\label{it:steptworqaoa}
Next, find a pair of vertices $(i,j)$ and a color $b\in\mathbb{Z}_k$ with the largest magnitude of $M_{i,j}(b)$ (breaking ties arbitrarily). Then impose the constraint
\begin{align}
x_j &= x_i+b \pmod k\ ,\label{eq:xjxiconstraint}
\end{align}
restricting the search space to the span of computational basis vectors~$\ket{x}$ associated with colorings  $x\in (\mathbb{Z}_k)^{n}$
satisfying~\eqref{eq:xjxiconstraint}. Observe that $\ket{\psi}$ 
has support on such basis states if and only if $M_{i,j}(b)=1$. To eliminate the variable $x_i$, the constraint~\eqref{eq:xjxiconstraint} is inserted into the cost function Hamiltonian as follows: use the identity
\begin{align}
\Pi_{i,j}(b)\Pi_{j,h}(a-b)&=\Pi_{i,j}(b)\Pi_{i,h}(a)
\end{align}
which holds for all $h\not\in \{i,j\}$ and all $a\in\mathbb{Z}_k$. Thus $\Pi_{i,h}(a)=\Pi_{j,h}(a-b)$ on the subspace satisfying the constraint. Replacing $\Pi_{i,h}(a)$ by $\Pi_{j,h}(a-b)$ in the cost function Hamiltonian for all $h\not\in\{i,j\}$ one gets a new Hamiltonian~$C'$ of the form~\eqref{eq:costfunctionhamiltonian} (up to an additive constant) acting on $n-1$~variables, i.e., $C'$ acts trivially on the $i$-th qudit. The cost function $C'$ is defined on a graph $G'$ with $n-1$ vertices obtained from $G$ by identifying the vertices $i$ and $j$.
\end{enumerate}
By construction, the maximum energy of $C'$ coincides with the maximum energy of~$C$ over the subset of assignments satisfying the constraints~\eqref{eq:xjxiconstraint}.  Since the new Hamiltonian $C'$ 
acts trivially on the qudit~$i$,  this qudit can be removed from the simulation. This completes the variable elimination step.

RQAOA$_\level$
executes several variable elimination steps in succession, eliminating one variable in each 
recursion
until the number of variables reaches a predefined cutoff value $n_{c}$. The remaining Hamiltonian then only depends on $n_{c}$ variables is minimized by a purely classical algorithm, for example brute-force search. An (approximate) solution~$x\in\mathbb{Z}_k^n$ of the original problem can then be obtained recursively by reconstructing eliminated variables using the constraints~\eqref{eq:xjxiconstraint}.

Since the introduction of RQAOA in~\cite{brakoeklietan}, other modifications of QAOA involving iterated rounding procedures have been proposed, see e.g., \cite[Section~V.A]{googlelowdepths}. In contrast to the variant discussed there, RQAOA's rounding procedure is deterministic and relies on rounding correlations between qudits rather than individual spin polarizations. 

Further variations of RQAOA have been proposed and studied in~\cite{egger2020warmstarting,SDP2020bridging}. There, the authors consider ``warm-starting'' the algorithm by beginning with a solution returned by an efficient classical algorithm (for the case of MAX-CUT, the Goemans-Williamson algorithm) instead of the standard product state. They provided numerical evidence that both QAOA and RQAOA can achieve better performance when supplemented with ``warm-starting'', making this a promising potential avenue for future investigations.

\section{Limitations of QAOA$_p$ applied to MAX-$k$-CUT}
\label{sec:upperbound}
Limitations for level-$p$ QAOA with $p=O(\log n)$ applied to MAX-CUT were obtained for ensembles of random $d$-regular graphs in~\cite{farhi2020quantum}. Here we show analogous results 
for QAOA applied to MAX-$k$-CUT for $k>2$.   Our main result is the following theorem. It shows that unless the level~$\level$ of QAOA grows at least logarithmically with~$n$, QAOA cannot be more than marginally better (for large $n$ and $d$) than randomly guessing a coloring, for most $d$-regular bipartite graphs. We denote by  $\cG^{bi}_{n,d}$ the ensemble of uniformly random $d$-regular bipartite graphs on $n$ vertices. For a graph $G$, we denote by  $\mathrm{MC}_k(G)=\max_{x\in\mathbb{Z}_k^n}C(x)$ the maximum $k$-cut of~$G$. Note that $\mathrm{MC}_k(G)=nd/2$ for any $d$-regular bipartite graph $G$. Let
$\alpha_\level^{\mathrm{MC}_k}(G)$ be the approximation ratio of MAX-$k$-CUT QAOA$_\level$ applied to $G$, that is, $\alpha_\level^{\mathrm{MC}_k}(G)=\bra{\psi(\theta_*))}C(G)\ket{\psi(\theta_*)}/\mathrm{MC}_k(G)$ where $\theta_*=(\beta_*,\gamma_*)\in(\mathbb{R}^k)^\level\times \mathbb{R}^\level$ is an optimal set of angles, and where $C(G)$ is the cost function Hamiltonian for $G$ defined by Eq.~\eqref{eq:costfunctionhamiltonian}.
\begin{theorem}\label{thm:kcut_bound}
 There is a constant $\zeta>0$ such that 
\begin{align}
\Pr_{G\sim \cG^{bi}_{n,d}}\left[\alpha_\level^{\mathrm{MC}_k}(G)\geq (1-1/k)+o_d(1)+o_n(1)\right] & \leq o(1)
\end{align}
for all degrees $d$ satisfying $d\geq \zeta$ and $d=o(\sqrt{n})$ and all levels $\level <  \frac{1}{2}\log_d n$.
\end{theorem}

Our approach generally follows the basic idea of~\cite{farhi2020quantum} and also exploits the fact that QAOA is a local, and furthermore uniform algorithm.
In contrast to the main result of~\cite{farhi2020quantum}, which is expressed as an  upper bound on the expected approximation ratio (over the choice of graph) achieved by  QAOA, Theorem~\ref{thm:kcut_bound} shows that the approximation ratio is upper bounded for almost all graphs in the considered ensemble.
Analogous to the setting of~\cite{farhi2020quantum}, where it is shown that the approximation ratio converges to~$1/2$ in the  limit of large~$d$, we show that the limiting value is the approximation ratio achieved by random guessing.  In the analysis of~\cite{farhi2020quantum}, 
the explicitly known expected  maximal cut size for random $d$-regular graphs is used. The analogous value for MAX-$k$-CUT is not known for random $d$-regular graphs.
 Instead, we rely on an upper bound on the size of the maximum $k$-cut for typical $d$-regular random graphs obtained in the analysis of an SDP-relaxation  from~\cite{cojaoghlan}. In our analysis, we pay special attention to the graph-dependence of the optimal QAOA angles when computing QAOA approximation ratios: here we make use of the fact that the figure of merit (the average energy) for different typical graphs differs only by a negligible amount for any fixed angles.

\newcommand*{\nt}{n_{T_{\level,d}}}

Given a graph $G = (V,E)$, let $(i,j)=e\in E$ be an edge in $G$. We write $N_\level(e)$ to denote the $\level$-local neighborhood of $e$, i.e., the subgraph induced by the set of all vertices $h$ such that $d(h,i)\le \level$ or $d(h,j)\le \level$. Let $T_{\level,d}$ denote the tree such that all vertices except the leaves are $d$-regular and such that there exists an edge $(i,j)$ such that all leaves are distance $\level$ from either $i$ or $j$.
Let $\nt(G)$ denote the number of edges $e$ in a graph $G$ such that $N_\level(e) \not\cong T_{\level,d}$. The key insight of~\cite{farhi2020quantum} is the 
fact that for $\level$ sufficiently small compared to~$n$, it holds with high probability that almost all $\level$-local neighborhoods of a random $d$-regular graph look like $T_{\level,d}$. This is expressed by the following lemma, where we denote by $\cG_{n,d}$ the uniform distribution over $d$-regular graphs on $n$ vertices.
\begin{lemma}\label{lem:tree}\cite{farhi2020quantum}
Let $n,d$ and $A \in (0,1)$ be given. Suppose that
\begin{align}
    (d-1)^{2\level} < n^A\ .\label{eq:dlevelA}
\end{align}
Then there exists some constant $A' \in (A,1)$ such that
\begin{align}
    \ExpE_{G\sim\cG_{n,d}}\left[\nt(G)\right] = O(n^{A'})\ .
\end{align}
In particular, by Markov's inequality, 
\begin{align}
\Pr_{G\sim \cG_{n,d}}[\nt(G) \ge n^B] = O(n^{A'-B}) = o(1)\qquad\textrm{ for all } B\in (A',1)\ .\label{eq:inequalitystatementlem}
\end{align}
Moreover, the same result also holds for  the ensemble $\cG_{n,d}^{bi}$  of random $d$-regular bipartite graphs.
\end{lemma}
The proof of Lemma~\ref{lem:tree} relies on the fact that the number of cycles of length $\ell$ in a random $d$-regular graph is Poisson distributed, allowing an upper bound to be established on the expected number of cycles below a certain length~\cite{WORMALD1981168}. This latter bound applies to both random $d$-regular graphs as well as random bipartite $d$-regular graphs. Both ensembles will be of relevance in the proof of Theorem~\ref{thm:kcut_bound}.

Let us now fix some $B \in (A',1)$. We will say that a graph $G$ on $n$ vertices is $T$-typical if $\nt(G) < n^{B}$. Lemma~\ref{lem:tree}  says that a random (possibly bipartite) $d$-regular graph~$G$ is $T$-typical with high probability.
Since QAOA is a local, and moreover uniform, algorithm, the performance of QAOA$_p$ on any given graph can be related to the performance of QAOA$_p$ on the induced $p$-neighborhood of each edge. Lemma~\ref{lem:tree} therefore suggests that to study QAOA$_p$ for generic $d$-regular graphs, it suffices to consider the behavior of QAOA$_\level$ on $T_{\level,d}$.

Let  now $C(G)$ be the cost function Hamiltonian for the MAX-$k$-CUT problem associated with $G$ (see Eq.~\eqref{eq:costfunctionhamiltonian}), and let $\psi(\theta)$ denote the level-$\level$ QAOA state defined by Eq.~\eqref{eq:psigammaqaoa}, where we use the shorthand $\theta=(\beta,\gamma)$ for the collection of all angles. 
 Then the expectation value for MAX-$k$-CUT QAOA$_p$ on $G=(V,E)$ can be written as
\begin{align}
\bra{\psi(\theta)}C(G)\ket{\psi(\theta)}&=\sum_{(i,j)\in E} \bra{\psi(\theta)} C_{i,j}\ket{\psi(\theta)}\ ,
\end{align}
where $C_{i,j}$ is the term in the Hamiltonian~\eqref{eq:costfunctionhamiltonian} acting non-trivially on both 
qudits $i$ and $j$. Note that due to the locality of the  QAOA unitary,  each of the individual terms  $\bra{\psi(\theta)} C_{i,j}\ket{\psi(\theta)}$ in this expression a function of the subgraph~$N_p(ij)$ and $\theta$.

Now, let $G$ be a $d$-regular $T$-typical graph. Each of the local terms $\bra{\psi(\theta)} C_{i,j}\ket{\psi(\theta)}$ is bounded above by $1$ since $C_{i,j}$ is a projection. Since for a $T$-typical graph~$G$, all but $\nt(G) < n^B$ of the edges satisfy $N_\level(e) \cong T_{\level,d}$, it follows that
\begin{align}
\bra{\psi(\theta)}C(G)\ket{\psi(\theta)} &= \left(\frac{nd}{2}-\nt(G)\right)C_T(\theta) + \sum_{(i,j):N_\level(ij)\not\cong T_{\level,d}} \bra{\psi(\theta)} C_{i,j}\ket{\psi(\theta)}\nonumber\\ 
&= \frac{nd}{2}C_T(\theta) + O(n^B)\ .\label{eq:universalbndctnb}
\end{align}
In this expression,  the function $C_T(\theta)= \bra{\psi(\theta)} C_{i',j'}\ket{\psi(\theta)}$ is the local term associated with any
one edge $(i',j')$ such that $N_\level(i'j') \cong T_{\level,d}$.  Importantly, the function $C_T(\cdot)$ depends on $(\level,d)$ only, i.e., it is universal for all $T$-typical graphs~$G$.

An important consequence is that the maximal expectation value achieved by level-$\level$ QAOA is roughly identical for all $T$-typical graphs $G$: if $G_1$ and $G_2$ are $T$-typical, then
\begin{align}
        \left|\max_\theta \bra{\psi(\theta)}C(G_1)\ket{\psi(\theta)} - \max_\theta \bra{\psi(\theta)}C(G_2)\ket{\psi(\theta)}\right| = O(n^B)\ .\label{eq:uniform_bound2}
\end{align}
This follows from the fact that 
\begin{align}
\big|\bra{\psi(\theta)}C(G_1)\ket{\psi(\theta)} - \bra{\psi(\theta)}C(G_2)\ket{\psi(\theta)}\big| = O(n^B)\qquad\textrm{ for any set of angles}~\theta\ \label{eq:gonetwotheta}
\end{align}
 as a consequence of~\eqref{eq:universalbndctnb}, and the inequality $\big|\|f\|_\infty - \|g\|_\infty\big| \le \|f-g\|_\infty$. Specializing
Eq.~\eqref{eq:gonetwotheta} to $\theta=\arg\max_{\theta }\bra{\psi(\theta)}C(G_1)\ket{\psi(\theta)}$
also shows that  a choice of optimal angles  for any given $T$-typical graph~$G_1$ will also be nearly optimal for all other $T$-typical graphs~$G_2$.

We are going to establish an upper bound on the performance of QAOA on $T$-typical graphs. To this end, we use the following upper bound on the typical size of the maximum $k$-cut of a graph~$G$ in the ensemble~$\cG_{n,d}$ of $d$-regular graphs.  Then the following holds:
\begin{theorem}\cite[Theorem~19]{cojaoghlan}\label{thm:sdp}
There exists constants $\lambda,\zeta>0$ such that
\begin{align}
    \Pr_{G\sim\cG_{n,d}}\left[\mathrm{MC}_k(G) \le \left(1-\frac{1}{k}\right)\frac{nd}{2} + \lambda n\sqrt{d}\right] \ge 1 - e^{-2n}
\end{align}
for all degrees $d$ satisfying $d \geq \zeta$ and $d=o(\sqrt{n})$.
\end{theorem}
This implies that for all degrees $d$ satisfying the hypotheses of Theorem~\ref{thm:sdp}, we have
\begin{align}
    \ExpE_{G\sim \cG_{n,d}}[\mathrm{MC}_k(G)] \le \left(1-\frac{1}{k}\right)\frac{nd}{2} + \lambda n\sqrt{d} + o(n)\ .\label{eq:kcut_bound}
\end{align}

\begin{lemma}\label{lem:typical_kcut}
Let $n,d$ satisfy the hypotheses of Theorem~\ref{thm:sdp} and let $\level$, $A$ be such that Eq.~\eqref{eq:dlevelA} holds. Then there exists a constant $B'\in (0,1)$ such that the following holds  for sufficiently large $n$: There exists a $T$-typical $d$-regular graph~$G$ satisfying
\begin{align}
    \mathrm{MC}_k(G) < \ExpE_{G\sim \cG_{n,d}}[\mathrm{MC}_k(G)] + n^{B'}.\label{eq:inequalitytoprove}
\end{align}
\end{lemma}
\begin{proof}
Let $\cT$ denote the set of $T$-typical graphs, and let us abbreviate \\$E = \ExpE_{G\sim \cG_{n,d}}[\mathrm{MC}_k(G)]$.
We show that a randomly chosen $d$-regular graph~$G\sim \cG_{n,d}$ satisfies the desired properties (i.e., $G\in\cT$ and Eq.~\eqref{eq:inequalitytoprove}) with non-zero probability. Indeed, we have  by the union bound
\begin{align}
\Pr_{G\sim \cG_{n,d}}\left[\cG\not\in\cT \textrm{ or } \mathrm{MC}_k(G)\ge E+n^{B'}\right]&\leq 
 \Pr_{G\sim \cG_{n,d}}[G \not\in \cT]+\Pr_{G\sim \cG_{n,d}}\left[\mathrm{MC}_k(G)\ge E+n^{B'}\right]\ .\label{eq:gnottmcg}
\end{align}
We show that sum on the rhs is strictly less than~$1$ for a suitable choice of $B'\in (0,1)$ in the limit $n\rightarrow\infty$, implying the claim. Indeed, by 
 Eq.~\eqref{eq:inequalitystatementlem}  of
Lemma~\ref{lem:tree}, there is some $A'\in (A,1)$ such that 
\begin{align}
\Pr_{G\sim \cG_{n,d}}[G \not\in \cT]&=O(n^{A'-B})\qquad\textrm{whenever }\qquad B\in (A',1)\ .\label{eq:ggnd}
\end{align}
On the other hand, we have with Markov's inequality
\begin{align}
\Pr_{G\sim \cG_{n,d}}\left[\mathrm{MC}_k(G)\ge E+n^{B'}\right]&=\frac{1}{1+n^{B'}/E}\leq 1-\frac{1}{2} n^{B'}/E 
\end{align}
by Taylor series expansion. Note that the latter is well-defined, i.e., $n^{B'}/E \rightarrow 0$ for $n\rightarrow\infty$ for all $B'\in (0,1)$, since we have $E = \Omega(n)$. This is because a max $k$-cut is at least as large as a max $2$-cut, and the expected max $2$-cut scales as $\Theta(n)$. Inserting~\eqref{eq:kcut_bound}, we conclude that
\begin{align}
\Pr_{G\sim \cG_{n,d}}\left[\mathrm{MC}_k(G)\ge E+n^{B'}\right]&\leq 1-\frac{1}{2}\left[\left(\left(1-\frac{1}{k}\right)\frac{d}{2} + \lambda\sqrt{d}\right)^{-1}n^{B'-1} + o(n^{B'-1})\right]\ .\label{eq:lastxn}
\end{align}
Combining Eqs.~\eqref{eq:gnottmcg},~\eqref{eq:ggnd},~\eqref{eq:lastxn} and comparing exponents of~$n$, we conclude that 
\begin{align}
\Pr_{G\sim \cG_{n,d}}\left[\cG\not\in\cT \textrm{ or } \mathrm{MC}_k(G)\ge E+n^{B'}\right]<1\quad&\textrm{ for sufficiently large }n \\ &\textrm{ whenever } A'-B<B'-1.
\end{align}
The Lemma therefore holds for any choice of $B'\in (1-B+A',1)$. 
\end{proof}
Now, let us fix a $T$-typical graph $G_*$ from Lemma~\ref{lem:typical_kcut} with some chosen constant $B'$. Then we have
\begin{align}
    \max_\theta\bra{\psi(\theta)}C(G_*)\ket{\psi(\theta)} \le \mathrm{MC}_k(G_*) <  \ExpE_{G\sim \cG_{n,d}}[\mathrm{MC}_k(G)] + n^{B'}.
\end{align}
The first inequality follows from the fact that the expected value of the MAX-$k$-CUT cost function for colorings returned by QAOA cannot exceed the actual maximum. From inequality~\eqref{eq:kcut_bound}, we then have the bound
\begin{align}
    \max_\theta\bra{\psi(\theta)}C(G_*)\ket{\psi(\theta)} \le \left(\left(1-\frac{1}{k}\right)\frac{d}{2} + \lambda\sqrt{d}\right)\cdot n + o(n).\label{eq:kcut_expectation_bound}
\end{align}
From the uniform bound Eq.~\eqref{eq:uniform_bound2} between the maxima of $T$-typical graphs, the same bound holds for
any $T$-typical graph~$G$ with the same $n$ and $d$ as $G_*$, that is
\begin{align} \left|\max_\theta \bra{\psi(\theta)}C(G)\ket{\psi(\theta)} - \max_\theta \bra{\psi(\theta)}C(G_*)\ket{\psi(\theta)}\right| = o(n). \end{align}

To obtain a bound on the approximation ratio, we can focus our attention on $d$-regular bipartite graphs, which are guaranteed to have a maximum $k$-cut of size $nd/2$ for any $k$. 
Since Lemma~\ref{lem:tree} applies equally to random bipartite graphs, so that a random bipartite $d$-regular graph is also $T$-typical with  probability~$1-o(1)$.
Letting $G$ be a $T$-typical bipartite graph with $n$ and $d$ satisfying the hypothesis of Theorem~\ref{thm:sdp}, the bound~\eqref{eq:kcut_expectation_bound} holds for $G$. Dividing through by $nd/2$, the approximation ratio $\alpha_p^{\mathrm{MC}_k}(G)$ is therefore bounded above by
\begin{align}
    \alpha_p^{\mathrm{MC}_k}(G) = \frac{2}{nd}\max_\theta\bra{\psi(\theta)}C(G)\ket{\psi(\theta)} \le \left(1-\frac{1}{k}\right) + o_d(1) + o_n(1).
\end{align}
This concludes the proof of Theorem~\ref{thm:kcut_bound}.

\section{Classical simulation of level-$1$ RQAOA}
\label{sec:classSimQAOA}

A polynomial-time classical algorithm for computing expectation values \\ $\bra{\psi(\beta,\gamma)}Z_jZ_k\ket{\psi(\beta,\gamma)}$ for level-$1$ QAOA states~$\ket{\psi(\beta,\gamma)}$ and the MAX-CUT cost function was given by Wang et al.~\cite{zhihui2018}. Our work~\cite{brakoeklietan} describes a generalized version of this algorithm applicable to any Ising-type cost function.

In this section, we consider the $k$-coloring cost function Hamiltonian~\eqref{eq:costfunctionhamiltonian} and more generally Hamiltonians of the form~\eqref{eq:costfunctionhamiltonian} and give a classical algorithm for computing expectation values $\bra{\psi}Z_u^rZ_v^s\ket{\psi}$, where $\ket{\psi}=\ket{\psi(\beta,\gamma)}$ is the level-$1$ QAOA state defined in Eq.~\eqref{eq:psigammaqaoa} and $Z_u$ is the generalized Pauli-$Z$ operator acting on a qudit~$u\in [n]$. 
The algorithm has runtime $O(k^5(d_u+d_v))$,
where $d_j$ is the degree of a vertex $j$.
For a constant number of colors~$k$, this scales at most linearly with~$n$. For constant-degree graphs, the computation requires a constant amount of time.

A natural application of our algorithm is finding optimal
angles $(\beta,\gamma)$ for level-$1$ QAOA states. Indeed, since the 
variational energy is a linear combination of the expected
values $\bra{\psi}Z_u^rZ_v^s\ket{\psi}$, it can be efficiently computed classically using our algorithm.
This eliminates the need to prepare the variational state $\ket{\psi}$
on a quantum device for each intermediate choice of the
angles $(\beta,\gamma)$ in the energy optimization
subroutine. 
 In fact, the technique discussed here can easily be adapted to also efficiently compute partial derivatives
$\frac{\partial}{\partial \gamma}\bra{\psi}Z_u^rZ_v^s\ket{\psi}$ and $\frac{\partial}{\partial \beta_j}\bra{\psi}Z_u^rZ_v^s\ket{\psi}$,
thus sidestepping the need of using a quantum device even if the latter is used to estimate the gradient of the cost function.
However, once the optimal angles $(\beta,\gamma)$ are found,
a quantum device would be needed to prepare the optimal variational state and perform a measurement to obtain
a classical solution of the problem. 
Since the final measurement is not used in the
recursive version of QAOA, our algorithm enables
efficient classical simulation of level-$1$ RQAOA
in its entirety.

\subsection{General algorithm for level-$1$ QAOA-type expectation values\label{sec:generalalgorithmlevelone}}
 Our formulation extends beyond our QAOA Ansatz for MAX-$k$-CUT and could be used in other contexts. Specifically, let $C$ be a $2$-local Hamiltonian on $n$ $k$-dimensional qudits given by
\begin{align}
C&=\sum_{1\leq u<v\leq n} C_{u,v}\ ,\label{eq:hamiltonianpqdef}
\end{align}
where we assume that the terms $C_{u,v}$ and $C_{u',v'}$ 
acting on different qudits $u<v$ and $u'<v'$ commute pairwise. 
We will also write $C_{v,u}\equiv \mathsf{SWAP} C_{u,v}\mathsf{SWAP}^\dagger$ for the interaction term $C_{u,v}$ acting on qudits $(v,u)$ (in this order) when $u<v$. 
 Let $\ket{+}\in\mathbb{C}^k$ be a qudit state.  We are interested in generalized QAOA states of the form
 \begin{align}
 \ket{\Psi(\beta,\gamma)} &= B(\beta)^{\otimes n} e^{i\gamma C}\ket{+}^{\otimes n}\label{eq:defpsibetagammap}
 \end{align}
 where $B(\beta):\mathbb{C}^k\rightarrow\mathbb{C}^k$ is a unitary on~$\mathbb{C}^k$ for any $\beta\in\mathbb{R}^k$. We assume that $B(-\beta)=B(\beta)^\dagger$ is the adjoint of $B(\beta)$. The states defined by Eq.~\eqref{eq:psigammaqaoa} are a special case.
 
Let $O$ be any two-qudit observable. Our goal is to compute the expectation value
 \begin{align}
 \mu_{u,v}(O)&= \bra{\psi(\beta,\gamma)}O_{u,v}\ket{\psi(\beta,\gamma)}\ ,\label{eq:muuvO}
 \end{align} 
 where the subscripts $u,v$ indicate the pair of qudits acted upon by $O$. Define the density matrix
 \begin{align}
 \rho_{u,v}&=\tr_{w\not\in \{u,v\}}\left(e^{i\gamma C}\proj{+}^{\otimes n}e^{-i\gamma C}\right)\ .
 \end{align}
 One can compute $\rho_{u,v}$ in time $O(n)$ by initializing the pair of qudits $u,v$ in the state
 $e^{i\gamma H_{u,v}}\ket{+}^{\otimes 2}$ and sequentially coupling $\{u,v\}$ to each of the remaining qudits~$w$ in the following way: qubit~$w$ is initialized in the state~$\ket{+}$ and then coupled to $\{u,v\}$ by the respective terms in the cost function~$C$, namely~$e^{i\gamma(C_{u,w}+C_{v,w})}$. Finally, $w$ is traced out. 

The algorithm for computing $\rho_{u,v}$ can be summarized as follows:
 \begin{mdframed}[
    linecolor=black,
    linewidth=2pt,
    roundcorner=4pt,
    backgroundcolor=gray!15,
    userdefinedwidth=\textwidth,
]
\begin{algorithmic}
 \STATE $\eta\gets e^{i\gamma C_{u,v}}\proj{+}^{\otimes 2}e^{-i\gamma C_{u,v}}$
\FOR {$w\in [n]\backslash \{u,v\}$} 
        \STATE $\eta \gets \cE_w(\eta)$
\ENDFOR
\RETURN $\eta$
\end{algorithmic}
\end{mdframed}
 where $\cE_w$ is a two-qudit quantum channel defined as
 \begin{align}
 \cE_w(\eta)&=\tr_2 \left[e^{i\gamma (C_{u,w}+C_{v,w})} (\eta\otimes \proj{+})e^{-i\gamma (C_{u,v}+C_{v,w})})\right]\ .\label{eq:ewcomput} 
 \end{align}
 The final state $\rho_{u,v}$ involves $n-2$ applications of~$\cE_w$ in general. 
 By restriction to $w$ which interact non-trivially with $\{u,v\}$, this can be reduced to fewer than $d_u+d_v$ applications, where $d_u$ is the degree of $u$ in the interaction graph associated with $C$ (i.e., $(u,w)$ is an edge if and only if $C_{u,w}\neq 0$). This improvement applies for example to the MAX-$k$-CUT cost function Hamiltonian~\eqref{eq:costfunctionhamiltonian} for a bounded-degree graph.

 Finally, the definition~\eqref{eq:defpsibetagammap} of $\ket{\psi(\beta,\gamma)}$ together with the 
 assumed commutativity of the terms $C_{u,v}$ in the Hamiltonian give that the expectation value~\eqref{eq:muuvO} can be computed according to
 \begin{align}
 \mu_{u,v}(O)&=\tr\left(\rho_{u,v}B(-\beta)^{\otimes 2}O_{u,v}B(\beta)^{\otimes 2}\right)\ .\label{eq:rhouvbbetacomputation}
 \end{align}
 This computation is illustrated in  Fig.~\ref{fig:manipulations}. To find the $k$-dependence of the overall complexity of this algorithm,  we need to consider the evaluation of the superoperators~$\cE_w$ defined by~\eqref{eq:ewcomput} and the expression~\eqref{eq:rhouvbbetacomputation} for the problem at hand.

\begin{figure}
\begin{center}
   \subfloat[The  expression of interest.]{\includegraphics[height=2cm, width=.6 \textwidth]{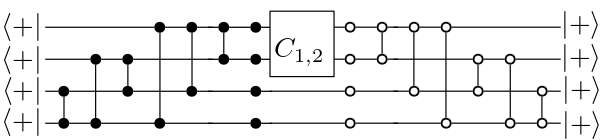}}\\ 
     \subfloat[Single-qubit unitaries $B(\beta)$ not acting on the qudits $\{1,2\}$ cancel. ]{\includegraphics[height=2cm, height=2cm, width=.6\textwidth]{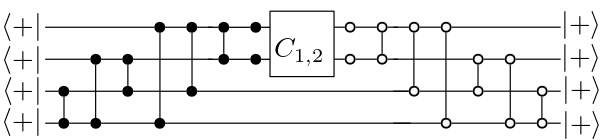}}\\
   \subfloat[The two unitaries
   $e^{\pm i\gamma C_{1,2}}$  can be moved to the left and right. 
   Unitaries    $e^{\pm i\gamma C_{p,q}}$ acting on qudits    $p,q\not\in \{1,2\}$ can be moved to the center. 
    ]{\includegraphics[height=2cm, width=.6\textwidth]{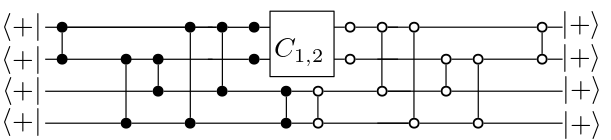}}\\
   \subfloat[where they cancel with their adjoints.]{\includegraphics[height=2cm, width=.6\textwidth]{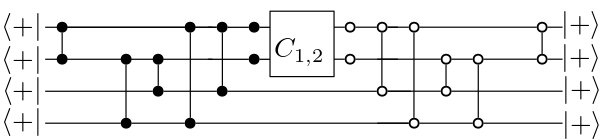}}\\
      \subfloat[Again using commutativity of the operators $C_{p,q}$, the remaining unitaries can be grouped
      into pairs of the form $e^{\pm i\gamma (C_{1,w}+C_{2,w})}$ where $w\not\in\{1,2\}$.]{\includegraphics[height=2cm, width=.6\textwidth]{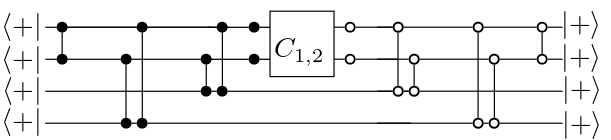}}\\
   \subfloat[The resulting expression can be interpreted as the result of applying a sequence of superoperators~$\cE_w$ to the state $e^{i\gamma C_{1,2}}\proj{+}^{\otimes 2}e^{-i\gamma C_{1,2}}$, and then taking the expectation of $B(-\beta)^{\otimes 2}C_{1,2}B(\beta)^{\otimes 2}$.]{\includegraphics[height=2cm, width=.6\textwidth]{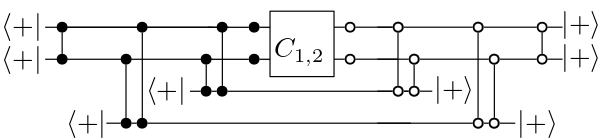}}      
   \caption{An expression 
   of the form~$\bra{\psi(\beta,\gamma)}C_{u,v} \ket{\psi(\beta,\gamma)}$.  Here we assumed that $n=4$ and $(u,v)=(1,2)$. Single-qudit operators of the form $B(\beta)$ are represented by circles, and their adjoints $B(\beta)^\dagger$ by filled circles. Similarly, gates of the form $e^{-i\gamma C_{p,q}}$ respectively $e^{i\gamma C_{p,q}}$ for qudits $p<q$ are represented by connecting lines between the corresponding qudits with empty respectively filled circles.}\label{fig:manipulations}
\end{center}
\end{figure}

\subsection{Simulating QAOA$_1$ for MAX-$k$-CUT}
Here we specialize the algorithm from Section~\ref{sec:generalalgorithmlevelone} to the MAX-$k$-CUT problem. For this problem,  $B(\beta)$ is given by Eq.~\eqref{eq:Deltadef}
whereas $C$ (cf.~Eq.~\eqref{eq:costfunctionhamiltonian})  is a diagonal $2$-local Hamiltonian commuting with $X^{\otimes n}$. The most general diagonal $2$-local Hamiltonian~$C$ commuting with~$X^{\otimes n}$ has the following form: Let $\{h(a)\}_{a\in\mathbb{Z}_k}$ be a family of $n\times n$-matrices satisfying
\begin{align}
\overline{h_{u,v}(r)}=h_{v,u}(r)=h_{u,v}(-r)\label{eq:huvradj}
\end{align}
for all $u,v\in [n]$ and $r\in\mathbb{Z}_k$.
Then $C$ can be written as
\begin{align}
C=\sum_{1 \leq u<v \leq n} C_{u, v}\ , \qquad C_{u, v}=\sum_{a \in \mathbb{Z}_{k}} h_{u, v}(a) Z_{u}^{a} Z_{v}^{-a}=\sum_{a \in \mathbb{Z}_{k}} h_{v,u}(a) Z_{v}^{a} Z_{u}^{-a}\ .\label{eq:hpqrdef}
\end{align}
From Eq.~\eqref{eq:huvradj} one gets $C_{v,u}=C_{u,v}$ and $C_{u,v}^\dagger=C_{u,v}$. 
Let us agree that $C_{u, u}=0$ for all $u$.  
 The Hamiltonian Eq.~\eqref{eq:costfunctionhamiltonian} takes the form~\eqref{eq:hpqrdef} with 
$h_{u,v}(a)=\frac{1}{k}\sum_{b\in\mathbb{Z}_k}J_{u,v}(b)\omega^{a b}$ 
since 
\begin{align}
    \Pi_{u,v}(b)=\frac{1}{k}\sum_{a\in\mathbb{Z}_k}\omega^{a b}Z_u^a Z_v^{-a}\ .\label{eq:ftpiuvb}
\end{align}
Note that
\begin{align}
\bra{a,b}e^{-i \gamma C_{p, q}}\ket{a,b}=\exp \left[-i \gamma \hat{h}_{p, q}(a-b)\right]\  \label{eq:matrixelementdiagonal} 
\end{align}
where the
Fourier transform $\hat{h}(b)$
of $h$ is defined as a hermitian $n\times n$ matrix
with entries
\begin{align}
\hat{h}_{p,q}(b)=\sum_{a \in \mathbb{Z}_{k}} h_{p, q}(a) \omega^{a b}\ ,\qquad \omega \equiv e^{2\pi i/k}\ .
\end{align}
A straightforward computation using~\eqref{eq:matrixelementdiagonal} then gives that the superoperator~$\mathcal{E}_w$ (cf. Eq.~\eqref{eq:ewcomput}) is equal to
\begin{align}
\mathcal{E}_w(\eta)=\frac{1}{k} \sum_{a \in \mathbb{Z}_{k}} D_w(a) \eta D_w(a)^{\dagger}
\end{align}
with diagonal Kraus operators
\begin{align}
 D_w(a)\ket{c,d}=\exp \left[-i \gamma \hat{h}_{u,w}(c-a)-i \gamma \hat{h}_{v,w}(d-a)\right]\ket{c,d}\ . \label{eq:diagop}
\end{align}
A general two-qudit mixed state can be classically described by a Hermitian matrix of size $k^2\times k^2$. Multiplying this matrix by a diagonal Kraus operator~$D_w(a)$ takes time~$O(k^4)$. Thus a single application of the quantum channel~$\cE_w$ can be simulated classically in time $O(k^5)$. In total, the matrix~$\rho_{u,v}$ can thus be computed in fewer than $O(k^5(d_u+d_v))$ steps, where $d_u$ and $d_v$ are the degrees of $u$ and $v$ in the interaction graph defined by~$C$.

Finally, we claim that  the expectation value 
$\mu_{u,v}(O)$ in \eqref{eq:rhouvbbetacomputation}
can be computed in time $O(k^5)$. Indeed, 
we have 
\[
\mu_{u,v}(O)=\mathrm{Tr}(\eta O_{u,v})
\quad \mbox{where} \quad 
\eta=B(\beta)^{\otimes 2} \rho_{u,v} B(-\beta)^{\otimes 2}.
\]
Let us write
\[
\rho_{u,v}=\sum_{a,b\in {\mathbb Z}_k} M(a,b)\otimes |a\rangle\langle b|
\]
for some $k\times k$ matrices $M(a,b)$.
Accordingly,
\[
(B(\beta)\otimes I) \rho_{u,v} (B(-\beta)\otimes I)
=\sum_{a,b\in {\mathbb Z}_k} M'(a,b)\otimes |a\rangle\langle b|\equiv \eta'
\]
where $M'(a,b)=B(\beta) M(a,b)B(-\beta)$ can be
computed in time $O(k^3)$ for each pair $a,b$ by
multiplying $k\times k$ matrices. 
Thus one can compute the matrix of $\eta'$ in
the $Z$-basis in time $O(k^5)$. Rewrite this matrix as
\[
\eta'=\sum_{a,b\in {\mathbb Z}_k}  |a\rangle\langle b|
\otimes L(a,b)
\]
for some $k\times k$ matrices $L(a,b)$.
Then 
\[
\eta = (I\otimes B(\beta))\eta' (I\otimes B(-\beta))
=\sum_{a,b\in {\mathbb Z}_k}  |a\rangle\langle b|
\otimes L'(a,b),
\]
where $L'(a,b)=B(\beta) L(a,b)B(-\beta)$ can be
computed in time $O(k^3)$ for each pair $a,b$ by
multiplying $k\times k$ matrices. 
Thus one can compute the matrix of $\eta$ in time $O(k^5)$,
as claimed.
Finally $\mu_{u,v}(O)=\mathrm{Tr}(\eta O_{u,v})$
can be computed in time $O(k^5)$.

\section{Comparison  of QAOA, RQAOA and classical algorithms} \label{sec:numericalresults}
With the algorithm proposed in Section~\ref{sec:classSimQAOA}, we have simulated level-$1$ QAOA and level-$1$-RQAOA on large problem instances of MAX-$k$-CUT. Our focus is on (approximate) $3$-colorability, i.e., $k=3$. 
To compare these hybrid algorithms to classical algorithms, we have also applied the best known efficient classical approximation algorithm (as measured by the approximation ratio) to each problem instance, see  Section~\ref{sec:intro} for an overview of these algorithms.

Concretely, we generate random $d$-regular, $3$-colorable connected graphs on $n$ vertices according to an ensemble~$\cG[d,n]$ defined as follows. Here we assume $d<2n/3$ to be even, and $n$ to be a multiple of~$3$. A graph $G$ is drawn from $\cG[d,n]$ as follows:
\begin{enumerate}
\item
Define a partition $[n]=V_1\cup V_2\cup V_3$ into three pairwise disjoint subsets of each size $|V_j|=n/3$ for $j=1,2,3$.
\item
For each pair $r<s$ with $r,s\in \{1,2,3\}$, generate a random bipartite $d$-regular graph with vertex set $V_r\cup V_s$ and bipartition $V_r:V_s$. This graph is generated by iteratively going through each  $v\in V_r$, and adding edges $(v,w_1),\ldots,(v,w_d)$ where $\{w_1,\ldots,w_d\}\subset V_s$ are chosen uniformly at random among those vertices in $V_s$ that (currently) have degree less than~$d$. If at some point during this process, no such vertices are available, the generation of the bipartite graph is restarted.

\item
 Check whether the obtained graph contains a complete graph on $3$ vertices  (i.e., a triangle) and is connected. If either of these properties is not satisfied, start over.
\end{enumerate}
By definition, a graph $G=(V,E)$ drawn from $\cG[d,n]$ is $3$-colorable and thus the cutsize of the maximum $3$-cut is equal to $C_{\max}=|E|=nd/2$. In particular, expected approximation ratios for QAOA$_1$ can be immediately computed from the expectation $\bra{\psi}C\ket{\psi}$ of the cost function Hamiltonian. Similarly, for any approximate coloring $x\in\mathbb{Z}_3^n$ produced by RQAOA$_1$ (or any algorithm, for that matter), the achieved approximation ratio is  $C(x)/C_{\max}$. We use  the efficient classical algorithm by Newman~\cite{newman18} for comparison, as it has the best approximation guarantee (worst-case bound) among all known efficient classical algorithms for $k = 3$, see Section~\ref{sec:intro}. Note that the theoretical lower bound $0.836008$  on the expected approximation ratio of the classical algorithm~\cite{newman18} used here matches the one established for the algorithm in~\cite{klerketal} in the case $k = 3$ and the algorithm in~\cite{goemansmax3cut}.

\paragraph{Details for implementation.}

We empirically observed that the energy landscape contains local maxima and stationary points which often prevent gradient descent in QAOA and RQAOA from finding the optimal solution. The presence of many points with a negligible gradient is a well-studied phenomenon under the name of 
barren plateaus~\cite{mcclean2018barren}. 
Several promising strategies for alleviating this problem have
been previously explored~\cite{otterbach2017unsupervised,yao2020policy,shaydulin2019multistart}.


For the specific case of level-1 RQAOA applied to MAX-$3$-CUT, the energy function to be optimized depends on four arguments, three of which can be efficiently computed if the fourth one is known. This makes grid search an attractive and cheap alternative to gradient descent as there is only one dimension to be searched. It allows us to sidestep potential convergence problems with gradient descent and to study large graphs. In more detail, the energy~$E(\beta,\gamma)=\bra{\psi(\beta,\gamma)}C\ket{\psi(\beta,\gamma)}$ of the cost function Hamiltonian is a function of $\beta=(\beta_0,\beta_1,\beta_2)\in\mathbb{R}^3$ and $\gamma\in\mathbb{R}$. However, as we show in Appendix~\ref{sec:optnograddescent}, for any fixed value $\gamma\in\mathbb{R}$, parameters $\beta\in\mathbb{R}^3$ maximizing the function $\beta\mapsto E(\beta,\gamma)$ can be found efficiently by  computing roots of a degree-$4$ polynomial and performing a binary search on the unit circle. This significantly reduces the dimensionality of the optimization problem: only an interval of values $\gamma\in\mathbb{R}$ needs to be searched. Because the function to be optimized further satisfies $E((\beta_0,\beta_1,\beta_2),\gamma)=E((-\beta_0,-\beta_2,-\beta_1)),-\gamma)$ as $C$ is self-adjoint and real, and $\overline{B((\beta_0,\beta_1,\beta_2))}=B((-\beta_0,-\beta_1,-\beta_2))$, it further suffices to restrict the grid search to the interval $\gamma\in [0,\pi)$.

We thus chose $50$ equidistant grid points $\gamma_1, \hdots, \gamma_{50}$ in the interval $[0, \pi]$ for $\gamma$. After finding the grid point $\gamma_s$ that minimizes the cost function (see Appendix~\ref{sec:optnograddescent}), we performed another refined grid search with $50$ additional equidistant grid points on the interval $[\gamma_{s-1}, \gamma_{s + 1}]$. 

In our numerical experiments  we find that  optimal angles for QAOA$_1$ for the considered random ensembles of graphs concentrate around certain values. This is in line with the analytical findings of~\cite{brandao2018fixed} for QAOA: there it is shown that for random ensembles of $3$-regular graphs, frequencies of certain local subgraphs concentrate, yielding a simple dependence of the figure of merit on the angles. Similar observations are also used in the proof of our Theorem~\ref{thm:kcut_bound}.

\paragraph{Choice of parameters.}
For each pair of parameters $(n,d) \in \{30, 60, 150, 300\} \times \{4, 6, 8,10\}$, we generated $20$~graphs from~$\cG[d,n]$.

\paragraph{Numerical results.} In Figures~\ref{fig:n30comparisonAppRatios},~\ref{fig:n60comparisonAppRatios},~\ref{fig:n150comparisonAppRatios} and~\ref{fig:n300comparisonAppRatios}, we illustrate obtained achieved approximation ratios of QAOA$_1$, RQAOA$_1$ and the best approximation ratio of the classical algorithm by Newman~\cite{newman18} over $100$ samples generated in the correlation rounding step (magenta, green and blue bars).
For Newman's algorithm, we additionally provide the empirical average and standard deviation of the samples, which are indicated through error bars. Also, the theoretically expected approximation ratio of~$\alpha_{\textrm{Newman}}=0.836008$ for Newman's algorithm is shown.

As expected, we find that RQAOA$_1$ significantly outperforms QAOA$_1$ for the considered family of graphs. Their performance deteriorates with increasing degree~$d$ (although RQAOA$_1$ is less susceptible to this). This is consistent with our no-go result (Theorem~\ref{thm:kcut_bound}), which applies to random $d$-regular graphs of sufficiently large (but also constant) degree~$d$. We note, however, that Theorem~\ref{thm:kcut_bound} does not immediately pertain to the numerical examples considered here because of the (non-explicit) lower bound on~$d$ and its asymptotic nature (as a function of~$n$).

More interesting is the comparison to Newman's classical algorithm. Being a randomized algorithm, it is natural to consider the empirical means and variance of the achieved approximation ratio (although in practice, only the best coloring among a constant number of runs would be used). An analytic understanding of the exact behavior of Newman's algorithm appears to be a difficult problem. As already pointed out by Goemans and Williamson in their seminal paper~\cite{goeWill95}, even computing the variance analytically appears to be challenging. 

In our numerical experiments we find that the variance of Newman's algorithm increases with $d$ and decreases with $n$. In cases where the variance is large, among the $100$~trials considered, we typically find an essentially optimal solution\footnote{Let us mention that
the $3$-coloring problem for dense $3$-colorable graphs admits an efficient algorithm~\cite{alonkahale97}, although this is not directly connected to the performance of Newman's algorithm.}. 
We note, however, that the empirical average approximation ratio appears to be close to the theoretical worst-case guarantee~$\alpha_{\textrm{Newman}}$.

 With the observations above we can distinguish two regimes when comparing the performance of RQAOA to Newman's algorithm:

In the first regime (see e.g., Figs.~\ref{fig:n60largevar} and~\ref{fig:n150largevar}), where Newman's algorithm has high variance, we find that the optimal solution returned is superior to the one provided by RQAOA. Nevertheless and perhaps surprisingly, the coloring output by RQAOA (which is a deterministic algorithm up to estimation errors for expectation values when used in practice) outperforms the average approximation ratio of Newman's algorithm.

In the second regime (see e.g., Figs.~\ref{fig:n30smallvar} and~\ref{fig:n300smallvar}), Newman's algorithm is strongly concentrated about its average. 
In this case, RQAOA$_1$ outperforms virtually  all instances returned by Newman's algorithm. This regime includes instances with a large number of vertices, indicating that RQAOA may be a useful heuristic algorithm for problems of practical interest.

Larger levels $p>1$ may further increase achieved approximation ratios of RQAOA. Unfortunately, exploring this may require 
new ideas or actual quantum devices. This is suggestive of RQAOA being a promising candidate for NISQ applications.

\begin{figure}[h!]
\begin{center}
    \subfloat[$n = 30$, $d = 4$]
    {\includegraphics[height = 5.5cm, keepaspectratio, width=1\textwidth]{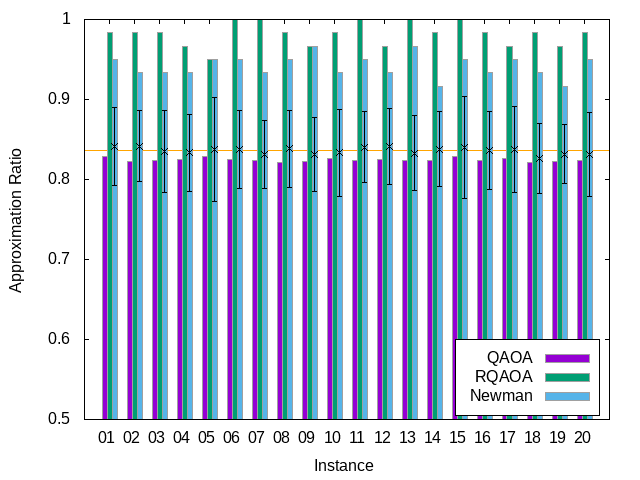}\label{fig:n30smallvar}}
    \subfloat[$n = 30$, $d = 6$ ]{\includegraphics[height = 5.5cm, keepaspectratio, width=1\textwidth]{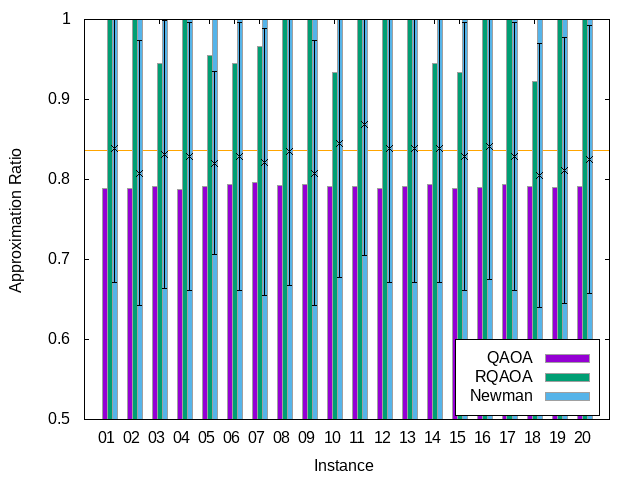}}\\
    \subfloat[$n = 30$, $d = 8$ ]{\includegraphics[height = 5.5cm, keepaspectratio, width=1\textwidth]{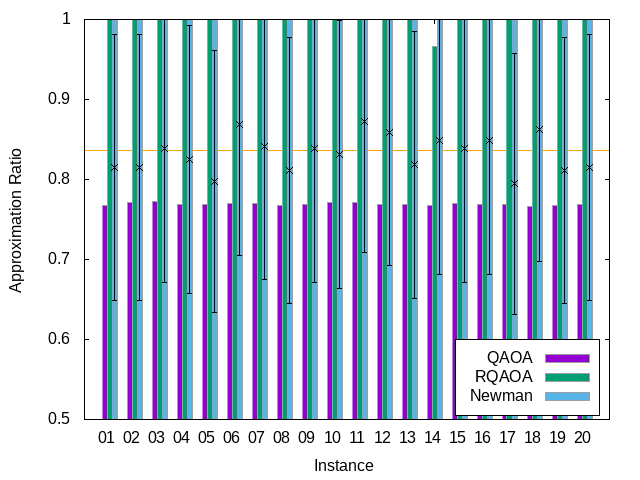}}
    \subfloat[$n = 30$, $d = 10$ ]{\includegraphics[height = 5.5cm, keepaspectratio, width=1\textwidth]{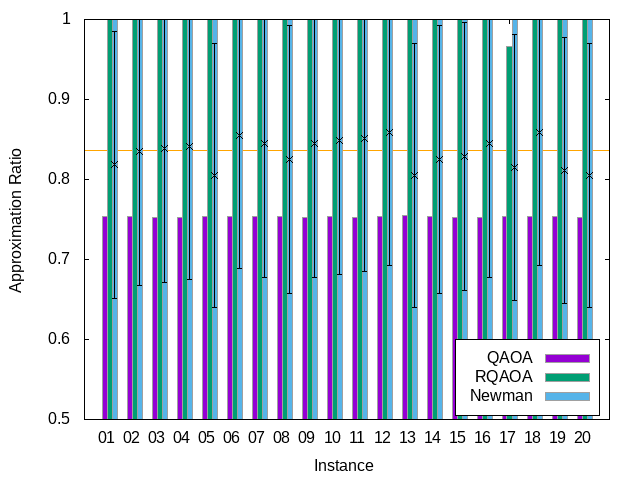}}
    \caption{Comparison of approximation ratios between QAOA$_1$, RQAOA$_1$ and the algorithm by Newman for MAX-$3$-CUT, where we took the best approximation ratio over $100$ samples of graphs with $n = 30$ vertices. The expected approximation ratio of Newman's algorithm is indicated by the horizontal line at $\alpha=0.836008$. For each graph, the empirical mean and standard deviation of Newman's algorithm are indicated through the error bars. }\label{fig:n30comparisonAppRatios}
\end{center}
\end{figure}

\newpage
\clearpage

\begin{figure}[h!]
\begin{center}
    \subfloat[$n = 60$, $d = 4$]
    {\includegraphics[height = 5.5cm, keepaspectratio, width=1\textwidth]{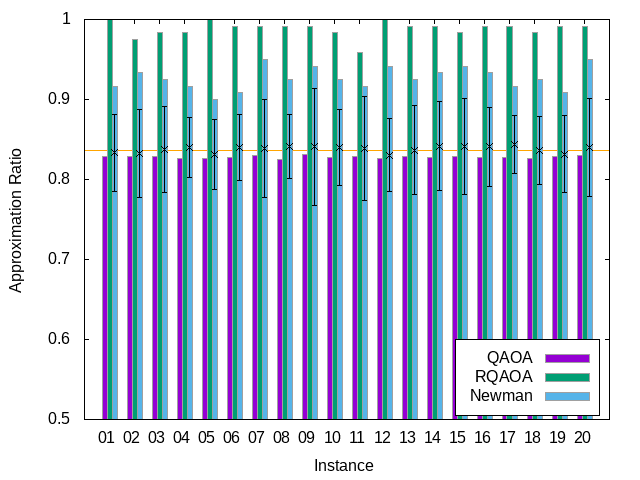}}
    \subfloat[$n = 60$, $d = 6$ ]{\includegraphics[height = 5.5cm, keepaspectratio, width=1\textwidth]{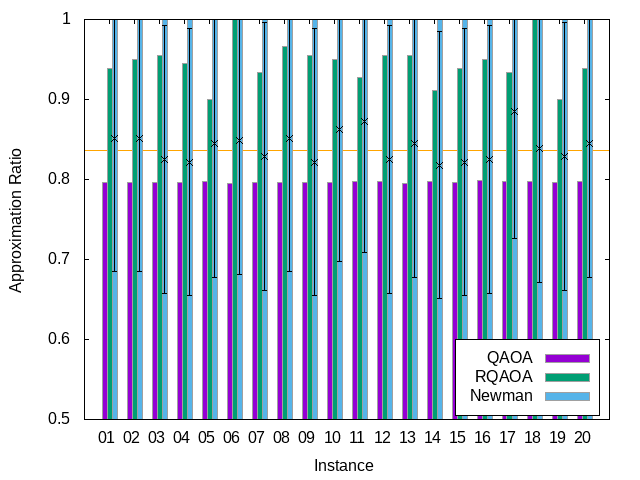}}\\
    \subfloat[$n = 60$, $d = 8$ ]{\includegraphics[height = 5.5cm, keepaspectratio, width=1\textwidth]{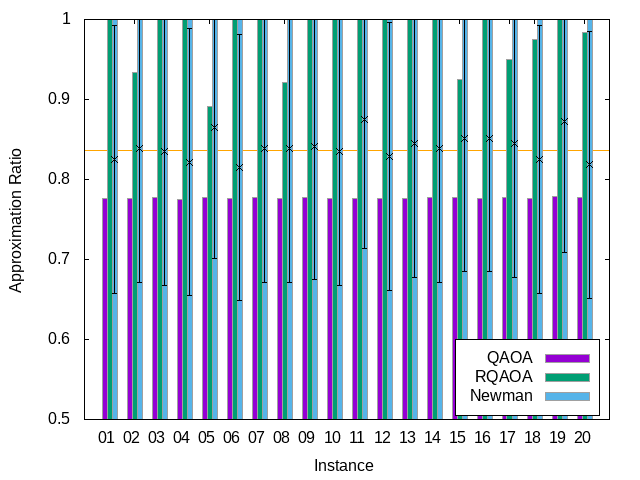}\label{fig:n60largevar}}
    \subfloat[$n = 60$, $d = 10$ ]{\includegraphics[height = 5.5cm, keepaspectratio, width=1\textwidth]{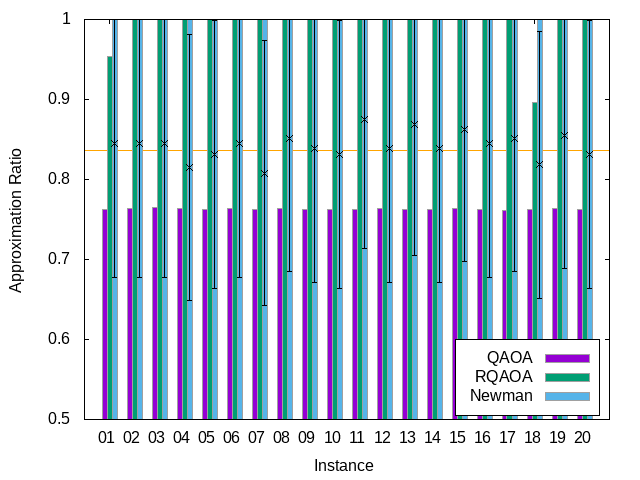}}
    \caption{Comparison of approximation ratios between QAOA$_1$, RQAOA$_1$ and the algorithm by Newman for MAX-$3$-CUT, where we took the best approximation ratio over $100$ samples of graphs with $n = 60$ vertices. The expected approximation ratio of Newman's algorithm is indicated by the horizontal line at $\alpha=0.836008$. For each graph, the empirical mean and standard deviation of Newman's algorithm are indicated through the error bars. }\label{fig:n60comparisonAppRatios}
\end{center}
\end{figure}

\newpage
\clearpage
\begin{figure}[h!]
\begin{center}
    \subfloat[$n = 150$, $d = 4$]
    {\includegraphics[height = 5.5cm, keepaspectratio, width=1\textwidth]{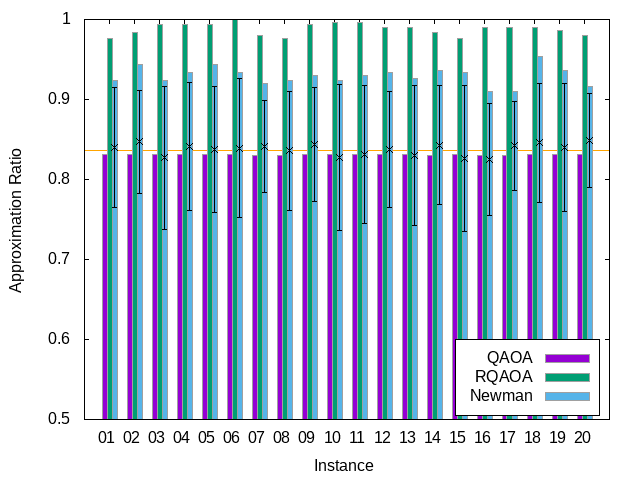}}
    \subfloat[$n = 150$, $d = 6$ ]{\includegraphics[height = 5.5cm, keepaspectratio, width=1\textwidth]{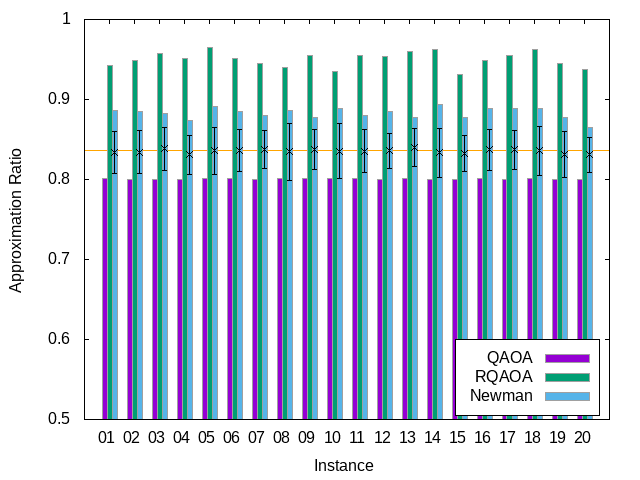}}\\
    \subfloat[$n = 150$, $d = 8$ ]{\includegraphics[height = 5.5cm, keepaspectratio, width=1\textwidth]{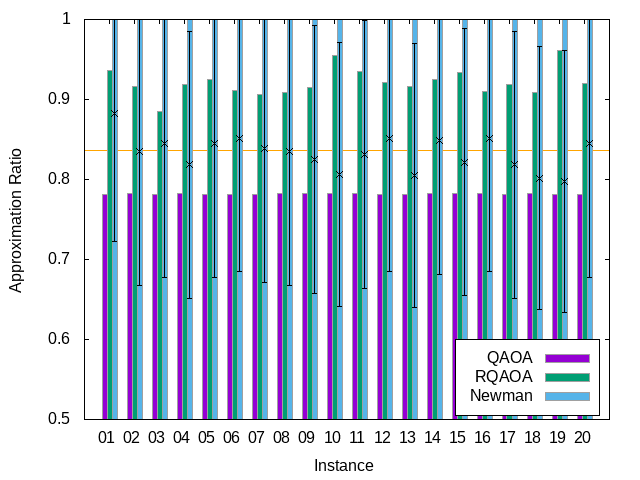}}
    \subfloat[$n = 150$, $d = 10$ ]{\includegraphics[height = 5.5cm, keepaspectratio, width=1\textwidth]{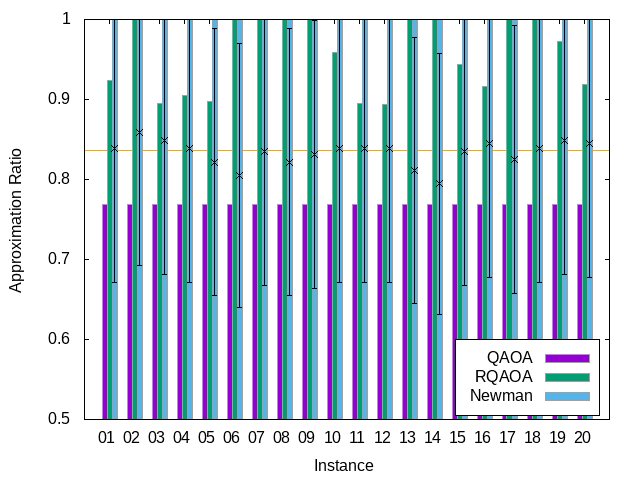}\label{fig:n150largevar}}
    \caption{Comparison of approximation ratios between QAOA$_1$, RQAOA$_1$ and the algorithm by Newman for MAX-$3$-CUT, where we took the best approximation ratio over $100$ samples of graphs with $n = 150$ vertices. The expected approximation ratio of Newman's algorithm is indicated by the horizontal line at $\alpha=0.836008 $. For each graph, the empirical mean and standard deviation of Newman's algorithm are indicated through the error bars. }\label{fig:n150comparisonAppRatios}
\end{center}
\end{figure}
\newpage
\clearpage

\begin{figure}[h!]
\begin{center}
    \subfloat[$n = 300$, $d = 4$]
    {\includegraphics[height = 5.5cm, keepaspectratio, width=1\textwidth]{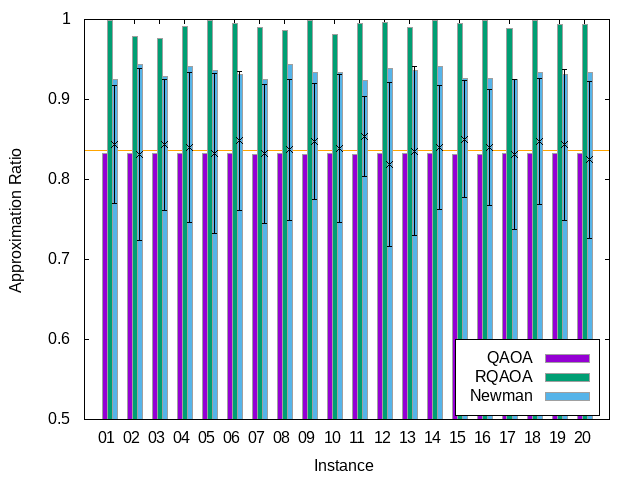}\label{fig:n300smallvar}}
    \subfloat[$n = 300$, $d = 6$ ]{\includegraphics[height = 5.5cm, keepaspectratio, width=1\textwidth]{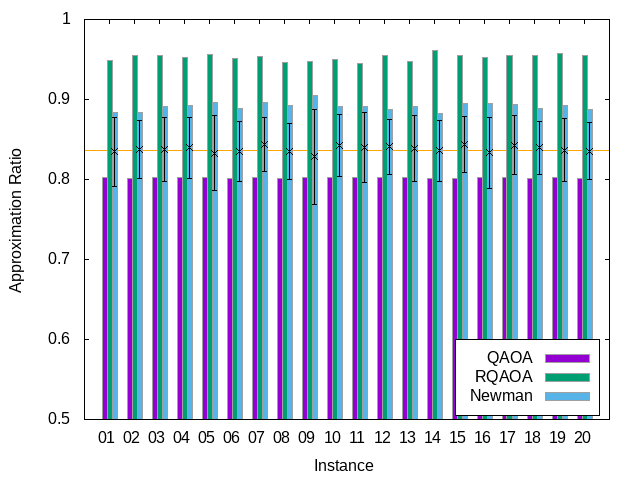}}\\
    \subfloat[$n = 300$, $d = 8$ ]{\includegraphics[height = 5.5cm, keepaspectratio, width=1\textwidth]{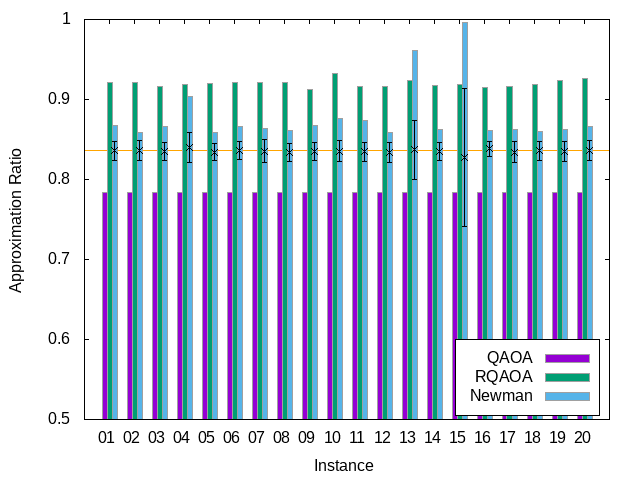}}
    \subfloat[$n = 300$, $d = 10$ ]{\includegraphics[height = 5.5cm, keepaspectratio, width=1\textwidth]{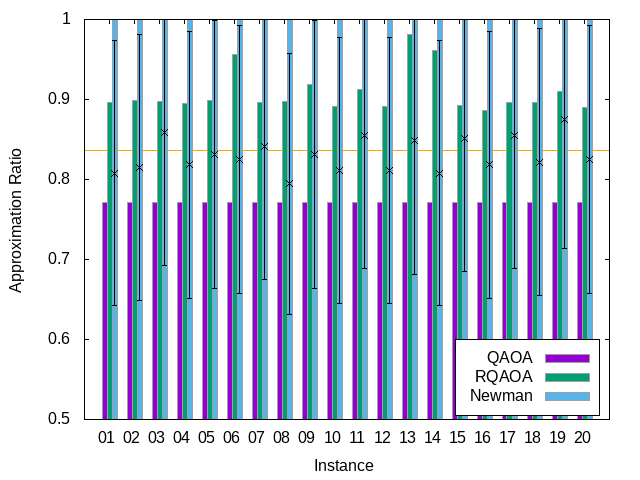}}
    \caption{Comparison of approximation ratios between QAOA$_1$, RQAOA$_1$ and the algorithm by Newman for MAX-$3$-CUT, where we took the best approximation ratio over $100$ samples of graphs with $n = 300$ vertices. The expected approximation ratio of Newman's algorithm is indicated by the horizontal line at $\alpha=0.836008 $. For each graph, the empirical mean and standard deviation of Newman's algorithm are indicated through the error bars. }\label{fig:n300comparisonAppRatios}
\end{center}
\end{figure}

\newpage

\section{Discussion and Outlook}
We formulated a variation of QAOA using qudits which is applicable to non-binary combinatorial optimization. We study its applicability by considering the MAX-$k$-CUT problem with $k>2$. We show that this version of QAOA shares certain limitations with the original proposal for qubits: its expected performance on random $d$-regular $k$-colorable graphs (for $d=o(\sqrt{n})$) approaches random guessing in the limit of large~$d$. We give an efficient classical simulation algorithm for its level-$1$ version. This method extends to problem  Hamiltonians with pairwise commuting $2$-local terms. The algorithm can also be used to simulate level-$1$ recursive QAOA (RQAOA), a hybrid classical-quantum algorithm designed to overcome the limitations of standard QAOA.

While efficient simulability precludes 
the possibility of a quantum advantage compared to the best efficient classical algorithm, it offers a crucial window into the performance of QAOA and RQAOA. We use our simulation algorithm to perform numerical experiments on graphs with up to $300$~vertices, obtaining approximation ratios for MAX-$3$-CUT achieved by level-$1$ QAOA and RQAOA. We compare these results to 
the best known efficient classical  algorithms for MAX-$3$-CUT, specifically the recent algorithm by Newman~\cite{newman18}. Our experiments show that RQAOA$_1$
(which significantly outperforms QAOA$_1$) is competitive with Newman's classical algorithm for generic $k$-colorable graphs. See Section~\ref{sec:numericalresults} for a detailed discussion.

Our observations suggest that  RQAOA may be a viable application for NISQ devices combined with efficient classical computation: it may be applicable to problem sizes of real-world relevance. Our results are indicative of the potential of RQAOA at larger levels~$p$, which  may no longer be simulable classically, and which hold promise. Future work may seek to establish performance guarantees for RQAOA$_p$ akin to the rigorous results available for existing classical algorithms. As an example, we proved in~\cite{brakoeklietan} that for the ``ring of disagrees'', RQAOA$_1$ achieves the optimal approximation ratio. A next step in the analysis of RQAOA would be to find more interesting classes of graphs for which achievability results can be established.

\textbf{Acknowledgments.} This work is supported in part by the Army Research Office under Grant Number W911NF-20-1-0014. The views and conclusions contained in this document are those of the authors and should not be interpreted as representing the official policies, either expressed or implied, of the Army Research Office or the U.S. Government. The U.S. Government is authorized to reproduce and distribute reprints for Government purposes notwithstanding any copyright notation herein. RK and AK gratefully acknowledge support by the DFG cluster of excellence 2111 (Munich Center for Quantum Science and Technology). ET acknowledges funding provided by DOE Award Number: DE-SC0018407 Quantum Error Correction and Spacetime Geometry and the Institute for Quantum Information and Matter, an NSF Physics Frontiers Center (NSF Grant PHY-1733907).

\appendix
\newpage
\section{MAX-$3$-CUT angle optimization without gradient descent}
\label{sec:optnograddescent}

First let us write the expectation value~\eqref{eq:rhouvbbetacomputation}
in the form that makes its dependence on $\beta$
more explicit. 
Recall from~\eqref{eq:Deltadef}
that for each $a\in \mathbb{Z}_k$, the $X$-eigenstate $\ket{\phi_a}\equiv Z^a\ket{+}$ is an eigenvector of the unitary $B(\beta)$ to eigenvalue~$e^{i\beta_a}$. Since the vectors $\{\ket{\phi_a}\}_{a\in\mathbb{Z}_k}$ form an orthonormal basis of~$\mathbb{C}^k$, one gets
\begin{align}
\mu_{u, v}(O)=\sum_{p, q \in \mathbb{Z}_k} e^{i\left(\beta_{p}+\beta_{q}\right)}\left\langle\phi_{p} \otimes \phi_{q}\left|\rho_{u, v} B(-\beta)^{\otimes 2} O_{u, v}\right| \phi_{p} \otimes \phi_{q}\right\rangle\ .
\end{align}
Consider the special case $O=Z^{r} \otimes Z^{-r}$. The identity $Z^{r}\left|\phi_{p}\right\rangle=\left|\phi_{p+r}\right\rangle$ gives 
\begin{align}
\mu_{u, v}\left(Z^{r} \otimes Z^{-r}\right)=\sum_{p, q \in \mathbb{Z}_{k}} e^{i\left(\beta_{p}+\beta_{q}-\beta_{p+r}-\beta_{q-r}\right)}\left\langle\phi_{p} \otimes \phi_{q}\left|\rho_{u, v}\right| \phi_{p+r} \otimes \phi_{q-r}\right\rangle\label{eq:zzexpectidentity}
\end{align}

For $(\beta,\gamma)\in\mathbb{R}^3\times\mathbb{R}$, let $E(\beta,\gamma)=\bra{\psi(\beta,\gamma)}C\ket{\psi(\beta,\gamma)}$ denote the energy in the level-$1$ QAOA state, where the cost function Hamiltonian~$C$ is given by Eq.~\eqref{eq:hpqrdef} with $k=3$.  Let $\gamma\in\mathbb{R}$ be fixed in the following. Let us write 
$E(\beta)=E(\beta,\gamma)$. Here we consider the problem of finding $\max_{\beta} E(\beta)$. 

With  Eq.~\eqref{eq:zzexpectidentity} we have 
\begin{align}
E(\beta)=\sum_{1 \leq p<q \leq n} \sum_{a, b, c \in \mathbb{Z}_{3}} h_{p, q}(c) e^{i\left(\beta_{a}+\beta_{b}-\beta_{a+c}-\beta_{b-c}\right)}\left\langle\phi_{a} \otimes \phi_{b}\left|\rho_{p, q}\right| \phi_{a+c} \otimes \phi_{b-c}\right\rangle\ .
\end{align}
 Clearly, the terms with $c=0$ do not depend on $\beta$ and the sum over $p, q$ gives $\operatorname{Tr}\left(\rho_{p, q}\right)=1 .$ For each $p<q$ the sum of all terms with $c=-1$ and the sum of all terms with $c=1$ are complex conjugates of each other. Indeed, these sums are expected values of $h_{p, q}(1) Z_{p} Z_{q}^{-1}$ and $h_{p, q}(1)^{*} Z_{p}^{-1} Z_{q}$. Here we noted that $h_{p, q}(-1)=h_{p, q}(1)^{*}$. Thus
\begin{align}
E(\beta) =\sum_{\substack{1 \leq p \\ <q \leq n}} h_{p, q}(0) +\sum_{\substack{a, b \\ \in \mathbb{Z}_{3}}} 2 \operatorname{Re}\left(h_{p, q}(1) e^{i\left(\beta_{a}+\beta_{b}-\beta_{a+1}-\beta_{b-1}\right)}\left\langle\phi_{a} \otimes \phi_{b}\left|\rho_{p, q}\right| \phi_{a+1} \otimes \phi_{b-1}\right\rangle\right)
\end{align}
Let $\bar{\beta} \equiv \beta_{0}+\beta_{1}+\beta_{2}$. Some simple algebra gives
\begin{align}
E(\beta)=C+\operatorname{Re} \sum_{a \in \mathbb{Z}_{3}} g_{a} e^{i\left(3 \beta_{a}-\bar{\beta}\right)}
\end{align}
where
\begin{align}
C=\sum_{1 \leq p<q \leq n} h_{p, q}(0)+2 \operatorname{Re} \sum_{a \in \mathbb{Z}_{3}} h_{p, q}(1)\left\langle\phi_{a} \otimes \phi_{a+1}\left|\rho_{p, q}\right| \phi_{a+1} \otimes \phi_{a}\right\rangle,
\end{align}
and
\begin{align}
g_{a}=2 \sum_{\substack{1 \leq p \\ <q \leq n}} h_{p, q}(1)\left\langle\phi_{a} \otimes \phi_{a}\left|\rho_{p, q}\right| \phi_{a+1} \otimes \phi_{a-1}\right\rangle+h_{p, q}(1)^{*}\left\langle\phi_{a} \otimes \phi_{a}\left|\rho_{p, q}\right| \phi_{a-1} \otimes \phi_{a+1}\right\rangle,
\end{align}
Define $\theta_{a}=3 \beta_{a}-\bar{\beta}$. Then $\theta_{0}+\theta_{1}+\theta_{2}=0.$ Thus it suffices to maximize a function
\begin{align}
F^{\prime}(\theta)=\operatorname{Re} \sum_{a \in \mathbb{Z}_{3}} g_{a} e^{i \theta_{a}}
\end{align}
over $\theta \in \mathbb{R}^{3}$ subject to a constraint
\begin{align}
\theta_{0}+\theta_{1}+\theta_{2}=0.
\end{align}
For fixed $\theta_{0}$ the maximum over $\theta_{1}$ can be computed analytically using the identity
\begin{align}
\max _{\theta_{1}} \operatorname{Re}\left(g_{1} e^{i \theta_{1}}+g_{2} e^{-i \theta_{0}-i \theta_{1}}\right)=\max _{\theta_{1}} \operatorname{Re}\left(e^{i \theta_{1}}\left(g_{1}+g_{2}^{*} e^{i \theta_{0}}\right)\right)=\left|g_{1}+g_{2}^{*} e^{i \theta_{0}}\right|.
\end{align}
The maximum is achieved at $\theta_{1}=-\arg \left(g_{1}+g_{2}^{*} e^{i \theta_{0}}\right) .$ Let $z \equiv e^{i \theta_{0}}$. It remains to maximize a function
\begin{align}
\begin{array}{c}
F^{\prime \prime}(z)=\max _{\theta \in \mathbb{R}^{3}} \quad F^{\prime}(\theta)=\operatorname{Re}\left(g_{0} z\right)+\left|g_{1}+g_{2}^{*} z\right| \\
\theta_{0}+\theta_{1}+\theta_{2}=0 \\
e^{i \theta_{0}}=z
\end{array}
\end{align}
over the unit circle $\{z \in \mathbb{C}:|z|=1\}$. Fix some $f \in \mathbb{R}$. One can easily check that conditions $F^{\prime \prime}(z)=f,|z|=1$ imply $p_{f}(z)=0,|z|=1,$ where $p_{f}(z)$ is a degree- 4 polynomial
\begin{align}
p_{f}(z)=&\left(-\frac{g_{0}^{2}}{4}\right) z^{4}+\left(g_{1}^{*} g_{2}^{*}+f g_{0}\right) z^{3}+z^{2}\left(\left|g_{1}\right|^{2}+\left|g_{2}\right|^{2}-f^{2}-(1 / 2)\left|g_{0}\right|^{2}\right) \\ &+z\left(g_{1} g_{2}+f g_{0}^{*}\right)-\frac{\left(g_{0}^{*}\right)^{2}}{4}
\end{align}
For a given $f$ one can analytically compute roots of $p_{f}(z),$ select roots lying on the unit circle, and check whether at least one of those roots $z$ satisfies $F^{\prime \prime}(z)=f .$ Thus one can check whether an equation $F^{\prime \prime}(z)=f$ has solutions $z$ on the unit circle. Now one can maximize $F^{\prime \prime}(z)$ over the unit circle using the binary search over $f$. Given the optimal value~$\theta_0=\arg\max_{\theta}F^{\prime\prime}(e^{i\theta})$,  one computes $\theta_{1}=-\arg \left(g_{1}+g_{2}^{*} e^{i \theta_{0}}\right), \theta_{2}=-\theta_{0}-\theta_{1},$ and solves a linear system $\theta_{a}=3 \beta_{a}-\bar{\beta}$ to find $\beta_{0}, \beta_{1}, \beta_{2}$.

\bibliographystyle{plain}

\bibliography{q}

\appendix
\end{document}